\documentclass[%
reprint,
superscriptaddress,
frontmatterverbose,
 amsmath,amssymb,
prb,
]{revtex4-1}

\usepackage{amsfonts,amsmath,amssymb,amsthm}

\usepackage{graphicx}
\usepackage{dcolumn}
\usepackage{bm}

\usepackage{hyperref}

\hypersetup{
	pdfborder={0 0 0}
	pdfnewwindow=true,      
	colorlinks=true,       
	linkcolor=blue,          
	citecolor=blue,        
	filecolor=blue,      
	urlcolor=blue,          
}

\usepackage{color}
\usepackage{amsmath}

\usepackage{lineno}
\makeatother

\newcommand{\be}{\begin{equation}}
\newcommand{\ee}{\end{equation}}

\newcommand{\ii}{{\mathrm{i}}}

\newcommand{\overbar}[1]{\mkern 1.5mu\overline{\mkern-1.5mu#1\mkern-1.5mu}\mkern 1.5mu}
\def\ket#1{\mathinner{|{#1}\rangle}}

\newtheorem{cor}{Corollary}
\newtheorem{lem}{Lemma}
\newtheorem{dfn}{Definition}
\newtheorem{propdfn}{Proposition/definition}

\usepackage{braket}

\newif\ifdraft
\drafttrue
\draftfalse

\begin{document}


\title{Worldline algorithm by oracle-guided variational autoregressive network}

\author{Zhifang Shi}
\affiliation{International Center for Quantum Materials, School of Physics, Peking University, Beijing 100871, China}

\author{Yuchuang Cao}
\affiliation{International Center for Quantum Materials, School of Physics, Peking University, Beijing 100871, China}

\author{Qiangqiang Gu}
\affiliation{International Center for Quantum Materials, School of Physics, Peking University, Beijing 100871, China}

\author{Ji Feng}\email{jfeng11@pku.edu.cn}
\affiliation{International Center for Quantum Materials, School of Physics, Peking University, Beijing 100871, China}
\affiliation{Collaborative Innovation Center of Quantum Matter, Beijing 100871,
China}
\affiliation{CAS Center for Excellence in Topological Quantum Computation, University of Chinese Academy of Sciences, Beijing 100190, China}


%

\date{\today}

\begin{abstract}
The variational autoregressive network is extended to the Euclidean path integral representation of quantum partition function. An essential challenge is adapting the sequential process of sample generation by an autoregressive network to a nonlocal constraint due to the periodic boundary condition in path integral. An add-on oracle is devised for this purpose, which accurately identifies and stalls unviable configurations as soon as they occur. The oracle enables rejection-free sampling conforming to the  periodic boundary condition. As a demonstration, the oracle-guided autoregressive network is applied to obtain variational solutions of quantum spin chains at finite temperatures with relatively large system sizes and numbers of time slicing, and to efficiently compute thermodynamic quantities.

\end{abstract}

\maketitle

\section{Introduction}

The Euclidean path integral\cite{Feynman53} representation of a quantum mechanical partition function creates a bridge between quantum and classical statistical mechanics, and has engendered continued developments in the theory of quantum many-body systems. For a lattice Hamiltonian $H$, the imaginary time $\tau$-axis [0, $\beta$] can be discretized into slices separated by $\Delta\tau$ and the Boltzmann measure becomes a product of $e^{-\Delta \tau H}$ evaluated on the spacetime lattice, with a suitable boundary condition. This reduces the quantum statistical problem to a classical one in $(d+1)$ dimensions, whose configuration can often be represented conveniently by worldlines.  Numerical techniques, often referred to as worldline algorithms, have been developed to evaluate the partition function based 
often on the Monte Carlo method,~\cite{Suzuki76, Suzuki77, Hirsch82, Evertz03} which is especially powerful for boson and quantum spin problems.

Neural network quantum states ansatz for quantum many-body systems has been recently introduced,\cite{Carleo602} in which  many-body wave functions coded by a neural network are optimized by variational Monte Carlo. The neural network quantum states ansatz has subsequently been extended to quantum ground states\cite{Carleo602,Saito17,Deng17,Glasser18,Vieijra19, Zheng19} and low lying excited states,\cite{Choo18} as well as open quantum many-body dynamics.\cite{Yoshioka19, Nagy19, Hartmann19, Vicentini19} Many of these approaches rely on Monte Carlo sampling, which then are subject to some of the inherent limitations of Monte Carlo algorithms, such as the lack of a universal cluster update scheme, critical slowing down and non-parallelizability of a Markov chain.

Recently, Wu \textit{et al} proposed a variational approach based on an autoregressive network to compute the partition function of a classical many-body system.\cite{Wu19} In this approach, the  distribution function of a classical system is factorized into a product of conditional probabilities that can be processed and optimized by an autoregressive network.
This variational autoregressive network (VAN) approach to classical statistical mechanics is free of the Markov chains, and consequently, highly parallelizable.
Furthermore, because its sample generation carries minimal correlation, the VAN method may avoid the critical slowing down.
VAN (or other generative models) has been adapted in some of quantum many-body problems mentioned above, including solving ground state wavefunctions,~\cite{Sharir20,Luo19,Hermann20,Choo19,carrasquilla2019,Hibat-Allah20,Morawetz21} simulating open quantum many-body dynamics,~\cite{Luo20} and solving variational density matrices.~\cite{Carrasquilla19,Liu21}

In this paper, we report an investigation of the VAN approach to quantum statistical problems.
Although VAN has been proved successful in solving ground state wavefunctions,~\cite{Sharir20,Luo19,Hermann20,Choo19,carrasquilla2019,Hibat-Allah20,Morawetz21} adapting VAN to path-integral Monte Carlo scheme to solve quantum statistical problems remains to be explored. For a quantum statistical problem at finite temperatures,  an essential obstacle to adapting VAN to a path integral  is the constraint that the worldlines must be periodic in $\tau$. This nonlocal constraint is not easily expressible in an autoregressive network, which parses and predicts configurations consecutively on a spacetime lattice. Here, this problem is overcome by an add-on oracle, which efficiently identifies and stalls dead ends among intermediate configurations in the sequential sample generation process. The oversight by the oracle ensures rejection-free sampling of the spacetime configurations, which are then parsed and optimized by the VAN to reach thermal equilibrium.  We illustrate our method using the $XXZ$ quantum spin chain model. It is found that thermodynamic quantities can be computed efficiently and accurately for antiferromagntic and $XY$-like regimes, establishing a technique for exploring problems in quantum many-body systems at finite temperatures.

\section{worldline algorithm by VAN}
Our variational autoregressive network to quantum statistical problems in the worldline representation is described in this section. We begin with a review of how the $XXZ$ quantum spin model is mapped to a classical one, which leads to the worldline representation of the quantum partition function. Then the VAN approach to this effective classical spin model is introduced. In particular, an oracle is devised to monitor the sequential and multi-pass process of sample generation by the VAN, in order to impose the periodic boundary condition in the imaginary time direction. Last, the action rescaling technique is introduced, which is shown to remedy the mode collapse problem.

\subsection{Path integral by checkerboard decomposition}
For concreteness, we consider a quantum spin-1/2 model on a 1-dimensional lattice with $N$ (even) sites, and the Hamiltonian involves only interactions between spin connected by nearest-neighbor bonds.
The $XXZ$ model with a Zeeman field $B$ 
\begin{equation}
	H=J\sum_{i} s_i^x s_{i+1}^x+s_i^y s_{i+1}^y +\lambda s_i^z s_{i+1}^z -B\sum_is_i^z
	\label{eq:xxz}
\end{equation}
breaks the $SU(2)$ symmetry of the Heisenberg model down to the $U(1)$ group, corresponding to rotation of every spin around the $z$-axis. Here  the in-plane coupling $J$ is used as the natural energy units,  and $\lambda=J_z/J$ is then the anisotropy factor. Periodic boundary condition is assumed. The spin operators $s_i^\alpha\; (\alpha=x,y,z)$ satisfy the commutation relations,
$
	[s_i^\alpha, s_{j}^\beta] = \ii \delta_{ij}\epsilon^{\alpha\beta\gamma} s ^\gamma.
$
The $U(1)$ symmetry means that the $z$-component of the total spin, $S^z=\sum s_i^z$, is conserved. To compute the partition function, $Z=\operatorname{Tr} e^{-\beta H}$, a path-integral representation is often used. Upon discretization of the imaginary time, the problem becomes a classic spin problem in 1+1 dimension.~\cite{Suzuki76, Suzuki77} We now briefly describe the process.

In the checkerboard decomposition,~\cite{Suzuki76, Suzuki77} 
$H$ can be divided into two separately commuting parts, involving even and odd bonds only, 
\begin{equation}
	H_{\text{even/odd}}= \sum_{i=\text{even/odd}} s_i^x s_{i+1}^x+s_i^y s_{i+1}^y +\lambda s_i^z s_{i+1}^z -h s_i^z,
\end{equation}
with $h$ is a Zeeman field in $z$-direction.
The partition function 
$
	Z=\operatorname{Tr} e^{-\beta H}
$
can be written as
\begin{equation}
	Z= \lim_{M\rightarrow\infty}\operatorname{Tr} \left[e^{-\Delta \tau(H_{\text{even}} + H_{\text {odd}})} \right]^M.
	\label{eq:3}
\end{equation}
$2M$ is the number of slices on the imaginary time axis, and $\Delta\tau = \beta/M$. Keeping $M$ finite and inserting the complete sets of basis that diagonalize $s^z$,  we find~\cite{Suzuki77,Evertz03}

\begin{equation}
\begin{aligned}
Z &= \sum\limits_{i_{1},\cdots,i_{2M}} \bra{i_1}e^{-\Delta\tau H_{\text{even}}} \ket{i_{2M}} \bra{i_{2M}}e^{-\Delta\tau H_{\text{odd}}} \ket{i_{2M-1}}  \\
&\qquad \qquad \qquad \cdots \bra{i_3}e^{-\Delta\tau H_{\text{even}}} \ket{i_2} \bra{i_2}e^{-\Delta\tau H_{\text{odd}}} \ket{i_1} \\
&= \sum_c S[c]=\sum_c \prod_{\pi_c} w(\pi_c)
\end{aligned}
\end{equation}

where $c=\{s(m,n)|m=0,\cdots,2M-1,n=0,\cdots,N-1\}$ is the spin configuration on the spacetime lattice and $w(\pi_c)$ is the plaquette weight to be discussed shortly.

\begin{figure}[ht]
	\includegraphics[width=80mm]{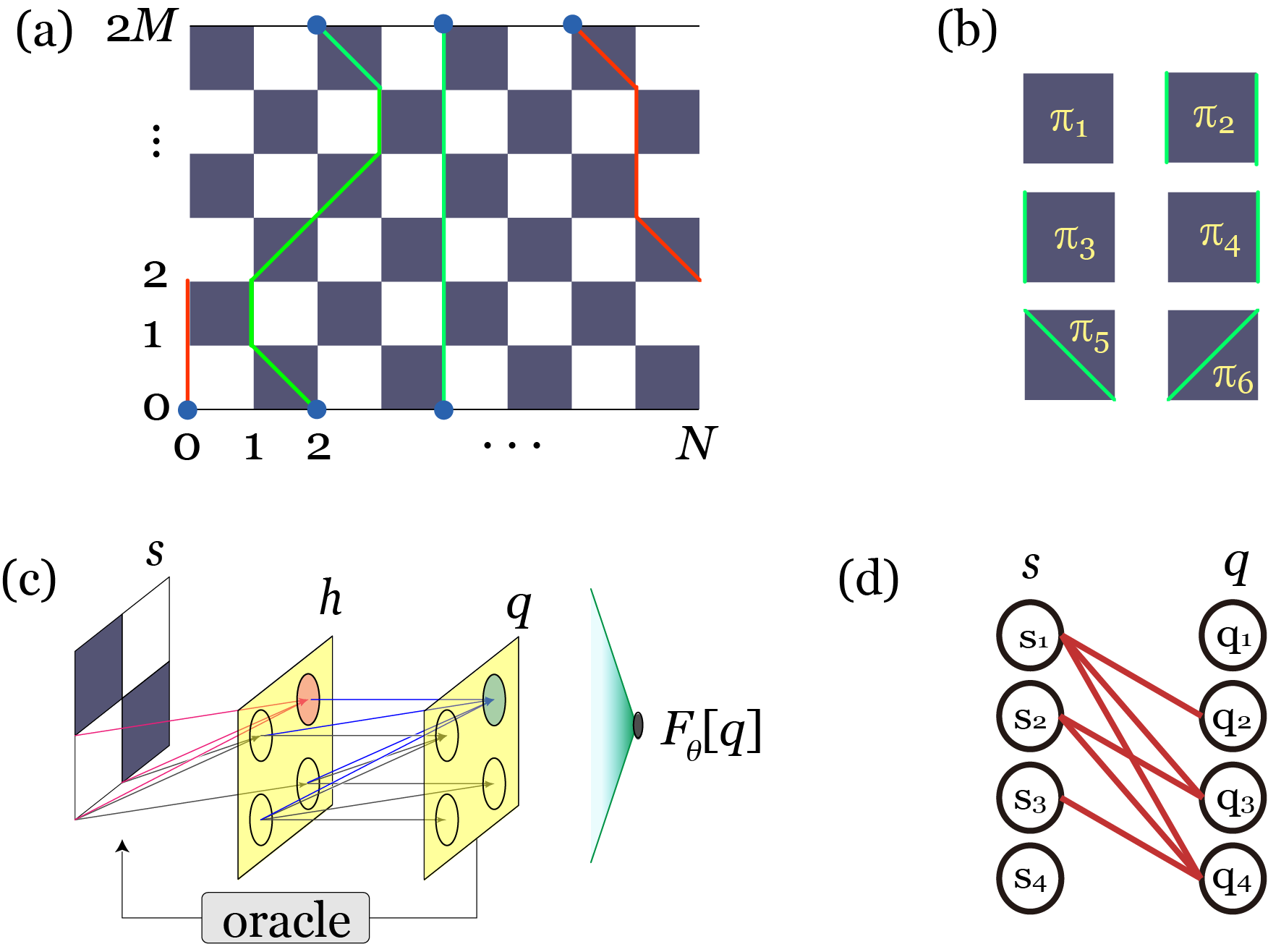}
	\caption{\label{fig:fig01}
		(a) The checkerboard representation of the quantum spin model. The green lines connecting blue dots at $\tau=0$ and $\beta$ are loops. The red line is an example of a worldline violating periodic boundary condition. (b) The plaquettes arising in the $XXZ$ model. (c) A schematic of the oracle-guided VAN with $M=1, N=2$. $s$ is a configuration to be generated, $h$ is a hidden layer, and $q$ is the last layer of neurons which produce $q_\theta$.  
		Each layer of the artificial network contains $2M\times N$ array of neurons. In the autoregressive architecture, succeeding layers are densely but not fully connected, which evidently mirrors the structure of the right-hand side in Eq. (\ref{eq:qtheta}). (d) The autoregressive connectivity between $s$ and $q$ layers, without showing the hidden layer. Here, the neurons representing spins and  $q_i$'s are shown as 1-dimensional arrays, though they correspond to the same network as in (c).
	}
\end{figure}

As a result of the checkerboard decomposition, the spacetime is partitioned into a $2M\times N$ checkerboard, as shown in Fig. \ref{fig:fig01}(a). For each configuration $c$, one can connect spin-up sites with non-intersecting and periodic paths, or loops, as illustrated in Fig. \ref{fig:fig01}(a). A loop starts at $\tau=0$ and ends at $\tau=\beta$ on the same lattice point, and owing to the spin conservation can only travel between adjacent time slices either vertically or diagonally across a dark square. Each spacetime lattice point the loops visit takes spin up state, whereas an unvisited lattice point is in the down spin state. A dark square decorated by the loop segments passing through along the vertical edges or the diagonals is referred to as a plaquette. Because of  the division of the $H$ into even and odd bond terms and the $U(1)$ symmetry, only the six types of plaquettes, as shown in Fig. \ref{fig:fig01}(b), need to be evaluated, whose weights are given by~\cite{Evertz03} 

\begin{eqnarray}
	\nonumber w(\pi_1) &=& e^{-\Delta\tau(\lambda -2h)/4},\\
	 w(\pi_2) &=& e^{-\Delta\tau(\lambda +2h)/4}, \\
	\nonumber w(\pi_3) =w(\pi_4) &=& e^{\Delta\tau \lambda /4} \cosh\tfrac{\Delta\tau}{2}, \\
	\nonumber w(\pi_5) =w(\pi_6) &=& e^{\Delta\tau \lambda /4} \sinh\tfrac{\Delta\tau}{2}.
\label{eq:w}
\end{eqnarray}

\subsection{Variational autoregressive network} 
Rewriting the partition function as $Z=\sum\exp(-\beta H_{\text{cl}})$, the spacetime action $S[c] = \beta H_{\text{cl}}$ leads to an energy function of the Gibbs measure, $H_{\text{cl}}$. Viewing $p[c]= e^{-\beta H_{\text{cl}}}/Z$ as a Gibbs distribution, the joint probability $p[c]$ can be factored into a product of successive conditional probabilities.\cite{Frey98, Goodfellow, Wu19} We now implement this factorization on an autoregressive network, as shown in Fig. \ref{fig:fig01}(c), 
\begin{equation}
	q_\theta(c) = \prod_{mn}q(s_{mn}=\uparrow|c_{mn}),
	\label{eq:qtheta}
\end{equation}
where $c_{mn}=\{s_{m'n'}|m'N+n'<mN+n\}$, and $q(s_{00}=\uparrow)=1/2.$\footnote{$q(s_{00}=\uparrow)=1/2$ if $Z_2$ symmetry is presented. otherwise $q(s_{00}=\uparrow)$ is determined by the network} The factorized $q_\theta(c)$ is optimized as close as possible to $p[c]$, by minimization of the free energy corresponding to $q_\theta(c)$. The network is not densely connected, as shown in Fig. \ref{fig:fig01}(c), as the neuron $q_i$ is only connected to neurons $s_{j<i}$. Because of this autoregressive connectivity, this network is referred to an \textit{autoregressive network}.\citep{Frey98,Goodfellow14,Uria16} As indicated in Fig. \ref{fig:fig01}(d), the overall connectivity between the spin configuration $s$ and $q_\theta$ is also autoregressive. 

In order to understand this viational process, we now describe the structure of VAN that factorizes $q_\theta(c)$. VANs are a type of generative model,\cite{Uria16,Wu19} which can generate samples with a variational distribution $q_\theta(c)$ parameterized by the autoregressive network. The model distribution $q_\theta(c)$ is optimized using the samples generated by VAN itself, which requires that the variational distribution $q_\theta(c)$ can be predicted from a given configuration $c$ and vice versa. Therefore we choose the \textit{autoregressive} network to parametrize $q_\theta(c)$. An illustration of the sample generation process by a simple autoregressive network, known as the fully visible sigmoid belief network,\citep{Frey98} is shown in Fig. \ref{fig:VAN}. The input layer on the left contains the sequentially generated sample $s=(s_1,s_2,s_3,s_4)$ as input. The output layer on the right produces conditional probabilities $q=(q_1,q_2,q_3,q_4)$, as shown in Eq. (\ref{eq:qtheta}), where $q_i$ corresponds to $q(s_{mn}=\uparrow|c_{mn})\in[0,1]$, the current predicted conditional probability  $s_{mn}=\uparrow$ is given the knowledge of spin states of all preceding sites. In each step, $q_i$ is determined by preceding spins $s_{j<i}$ (red neurons in $s$ layer), and one spin (blue neuron) $s_i$ is generated according to conditional probability $q_i$. A configuration $c$ is generated with a total of $2M\times N$ passages of the network in a row-major order over the spacetime lattice (Fig. \ref{fig:fig01}(a)). It should be remarked that only on the first row $(\tau=0)$ the spins are generated freely by the network. For all subsequent time slices, the spin conservation on each plaquette is coded into our VAN.

\begin{figure}[ht]
    \centering
	\includegraphics[width=80mm]{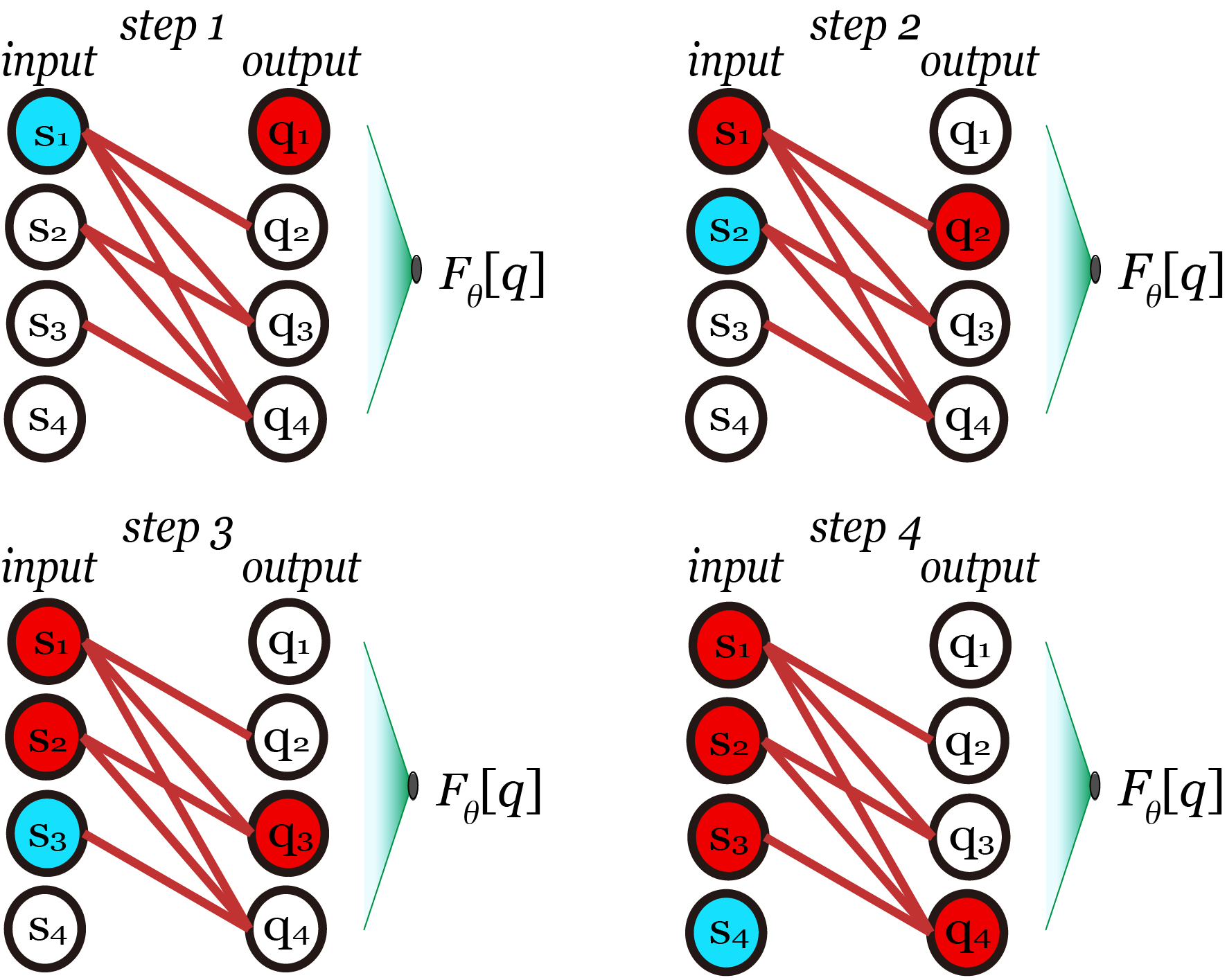}
	\caption{\label{fig:VAN}
	An illustration of how a VAN generates a configuration of 4 spins in 4 steps as discussed in the text. The input layer on the left is given the sequentially generated sample $s$ as input. The output layer on the right produces conditional probabilities $q$. In each step, $q_i$ is determined by preceding spins $s_{j<i}$ (red neurons in $s$ layer), and one spin (blue neuron) $s_i$ is generated according to conditional probability $q_i$.
	}
\end{figure}

Thus, the autoregressive network constitutes a surrogate model to produce the distribution for sampling the spin configuration, which can be used to compute the free energy
\begin{equation}
	F[q_\theta ]= \frac{1}{\beta}\sum_c q_\theta S[c]+ q_\theta \log q_\theta.
	\label{eq:fe}
\end{equation} 
The best approximant $q_\theta$(c)  to $p(c)$ is obtained by minimizing the free energy $F[q_\theta]$ with respect to the variation of network parameters $\theta$, in a standard batched stochastic minimization process.\cite{Goodfellow} This procedure formally describes a variational approach to the quantum statistical problem. 

\subsection{Oracle}
In conventional loop algorithms, a spacetime lattice configuration is subject to periodic updates in every Monte Carlo step~\cite{Suzuki76, Suzuki77, Hirsch82, Evertz03}   A key obstacle to the VAN for the quantum statistical problem outlined above lies with the fact that periodic boundary condition is a nonlocal constraint and consequently difficult to implement in a sequential process of sample generation. Thus a na\"ive implementation inevitably leads to a large number configurations violating the boundary condition, which we call the dead ends, and consequently has low sampling efficiency and staggeringly slow learning rate.
To overcome this problem, we introduce an add-on oracle\footnote{In context of computational complexity theory, an oracle is a black box attached to a Turing machine, which solves a certain decision problem.} to VAN, which is designed to predict and eliminate dead ends far before they occur, which we now describe.

\begin{figure}[ht]
	\includegraphics[width=80 mm, height=56.5 mm]{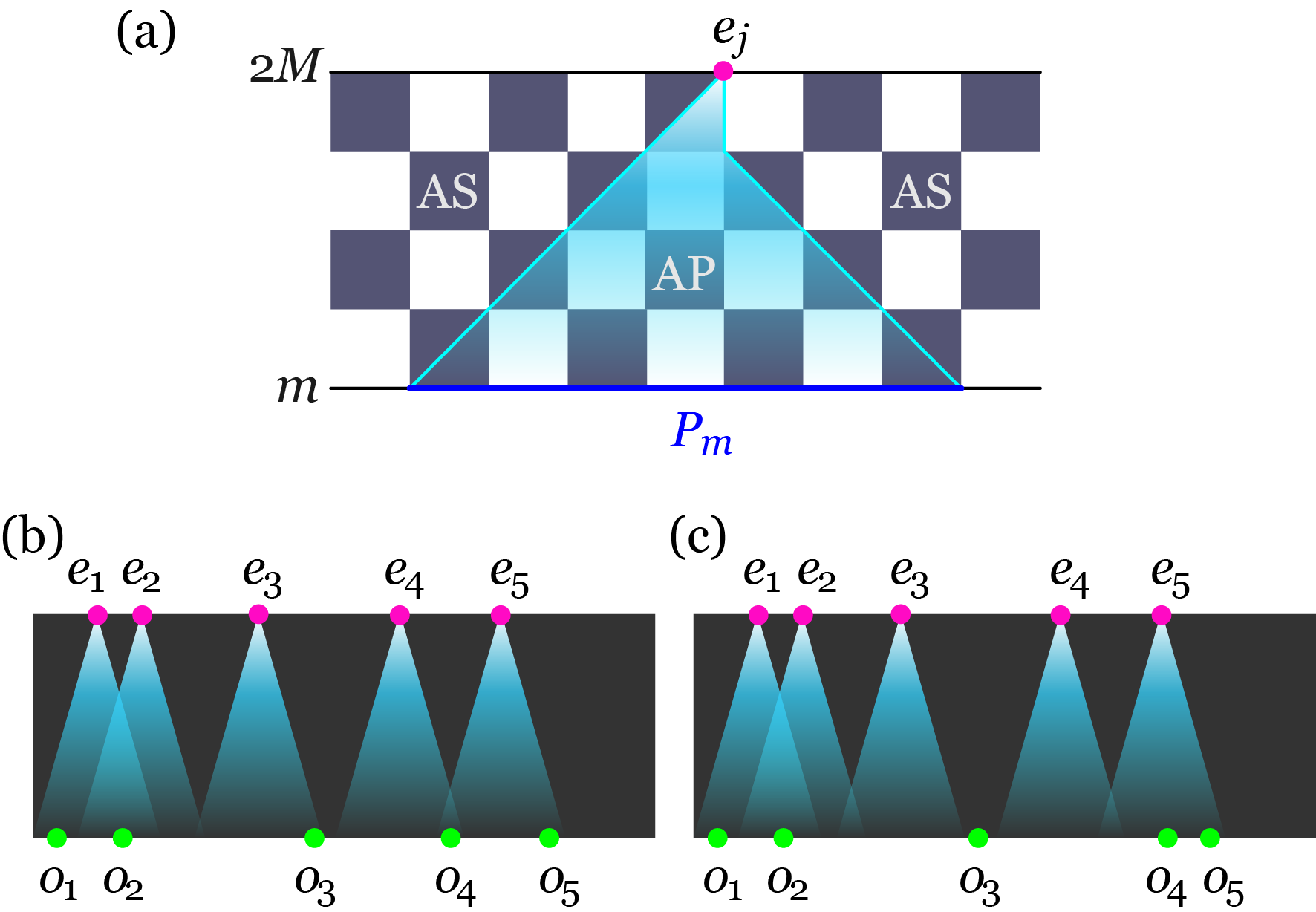}
	\caption{\label{fig:fig02}
	(a)A single lightcone emanating from one end point is shown. The region subtended by the lightcone (cyan lines) is in the absolute past (AP) of $e_j$, and everywhere else is absolutely separated (AS). 
	(b) and (c) show a spacetime configuration with five loops, with end points $e_j$, and points $o_j$ on a time slice at an intermediately step of VAN sampling. The dark regions are absolutely separated from all end points. 
	}
\end{figure}

At the beginning of a configuration generation, the spin configuration of time slice $\tau=0$ is decided stochastically by VAN, giving  a set of starting points of the loops which will be returned to at $\beta$. Call the set of end points $\{e_i=(x_j, \beta)\}$. For an end point $e_i$, the left- and right-most paths leading up to it form a {\it lightcone}, in the sense that all points bounded between this pair of paths can find a way to end at $e_i$ and are hence in the absolute past of $e_i$, and all other points are absolutely separated from $e_i$. The lightcone for a single end point is shown in Fig. \ref{fig:fig02}(a), of which the left- and right-most path are constructed by the greedy algorithm, heading diagonally across the dark squares. Shown in Fig. \ref{fig:fig02}(b) and (c) are examples with multiple end points and correspondingly multiple lightcones. We employ a width-first greedy algorithm\cite{Cormen01} to construct lightcones for every $e_i$, with which we can assess the viability of an intermediate configuration.

The oversight of VAN sampling by the oracle kicks in only for the last $
	m_s = \operatorname{min}(M+1, N/2+1)
$
steps, instead of at the very beginning of a sample generation. This is ensured by a lemma: a set of legit paths can always be found connecting $k$ pairs of grid points located separately on two time slices $N/2+1$ apart. The proof is elementary but lengthy, so we leave it to Appendix A.  Every time the spins on a time slice are created by the VAN, the oracle is queried for the viability of the current configuration by inspecting simply the location of spin up sites $o_j$ on current time slice, relative to the lightcones.
Shown in Fig. \ref{fig:fig02}(b) is a viable intermediate configuration where every intermediate point $o_j$ can be uniquely assigned to one of the lightcones emanating from $e_j$'s.  The situation of Fig. \ref{fig:fig02}(c) is considered unviable for two reasons: $o_3$ is absolutely separated from  all end points, and the lightcones of $e_3$ and $e_4$ are empty. An intermediate configuration deemed by the oracle to dead-end will be rejected right away. 

Incidentally, the oracle not only enables rejection-free sampling, it also alleviates partly the mode collapse problem commonly plaguing generative adversarial networks.\citep{Goodfellow14, Che17} Mode collapse happens when a VAN fails to be ergodic, only producing a small portion of configurations. For a comparison, we also perform a sampling without the oracle where dead ends are allowed, but with a large penalty to the loss function. As shown in Fig. \ref{fig:fig3}(a) for a 16-site antiferromagnetic ($J>0$) antiferromagnetic chain ($\beta |J|=1$), the penalty approach shows a sequence of plateauing in the variational process. In contrast, the oracle-guided VAN shows much faster convergence toward the exact free energy without visible plateaus. This is because in rejection-free sampling the network is better focused on the specialized learning goal it is designed for, whereas in the penalty method the early training is biased toward prioritizing eliminating the dead ends by generating ferromagnetic configurations.

\subsection{VAN sampling}

A worldline configuration $c$ is generated with a total of $2M\times N$ passages of the network in a row-major order. Only on the first row $(\tau=0)$ the spins are generated freely by the network; for all subsequent time slices, the continuity of worldline and periodic boundary condition along the imaginary time axis is coded into our VAN. This is done by masking the output conditional probabilities of VAN: if a random spin direction is deemed to dead end, it will be forced to another direction while the VAN output conditional probability of this spin is being masked. 

As $S^z=\sum s_i^z$ is conserved, the worldline in Fig. \ref{fig:fig01}(a) is continuous. This continuity is easy to implement
in the VAN sampling procedure. In the XXZ model case, the two spins on the top corners of a dark square are restricted by spin conservation: 
if bottom corners spins are $\ket{\uparrow\uparrow}$ / $\ket{\downarrow\downarrow}$, top corners spins must be $\ket{\uparrow\uparrow}$ / $\ket{\downarrow\downarrow}$; 
if bottom corners spins are $\ket{\uparrow\downarrow}$ or $\ket{\downarrow\uparrow}$, top corners spins can be either $\ket{\uparrow\downarrow}$ or $\ket{\downarrow\uparrow}$. In the latter case, only the first spin is randomly sampled by VAN, and the VAN output conditional probability for the second spin is masked($q_i$ forced to 0 or 1). 

As discussed earlier, an oracle is turned on for the last $\operatorname{min}(M, N/2)$ rows to prevent dead ends. Due to the space periodic boundary condition, a worldline configuration could have a nonzero winding number(here we define winding number as the number of times of worldlines cross the spatial boundary). When we sample the last $\operatorname{min}(N/2, M)$ rows, the $k$ worldlines of the last $\operatorname{min}(N/2, M)$ rows could have a winding number $w\in[-k,k]$, so the $k$ intermediate points $\{o_1,\cdots,o_k\}$ and the $k$ lightcones of end points $\{e_1,\cdots,e_k\}$ have $2k+1$ potential matches

\begin{equation} 
\begin{pmatrix}     e_1-N & \cdots &e_k-N\\  
    o_1 & \cdots & o_k \end{pmatrix}
,    
 \cdots , 
\begin{pmatrix}     e_1+N &\cdots & e_k+N\\  
    o_1 & \cdots & o_k \end{pmatrix}.\nonumber
\end{equation}
During the sampling procedure, all $2k+1$ matches are considered to ensure ergodic VAN sampling.

\subsection{Action rescaling}
As shown in Fig. \ref{fig:fig3}(a), the oracle alleviates partly the mode collapse problem. However, it is also observed that the oracle alone does not eliminate mode collapse  during the variational process completely. Therefore, we also introduce a action rescaling 
\begin{equation}
	S_{\text {rescaled}} = S_{\text{min}}+(1-r^{\text{epoch}})(S-S_{\text{min}}),
\end{equation}
in which $0<r<1$ is the strength of rescaling, $S_{\text{min}}$ is the smallest action value within the current epoch. This rescaling narrows the distribution of reward signal initially. This evens out the distribution of initial VAN predictions and keeps the samples from being attracted to locally plausible configurations.
 As shown in Fig. \ref{fig:fig3}(b), for a 32-site Antiferromagnetic chain with little action rescaling ($r=0.9$), mode collapse shows up in the oracle-guided VAN solution. But as soon as we  increase $r$, mode collapse disappears. We also remark that large $r$ value leads to slower early descent of the loss function, but better approach to the exact free energy later in the minimization.

\begin{figure}[ht]
	\includegraphics[width=80 mm, height=38.5 mm]{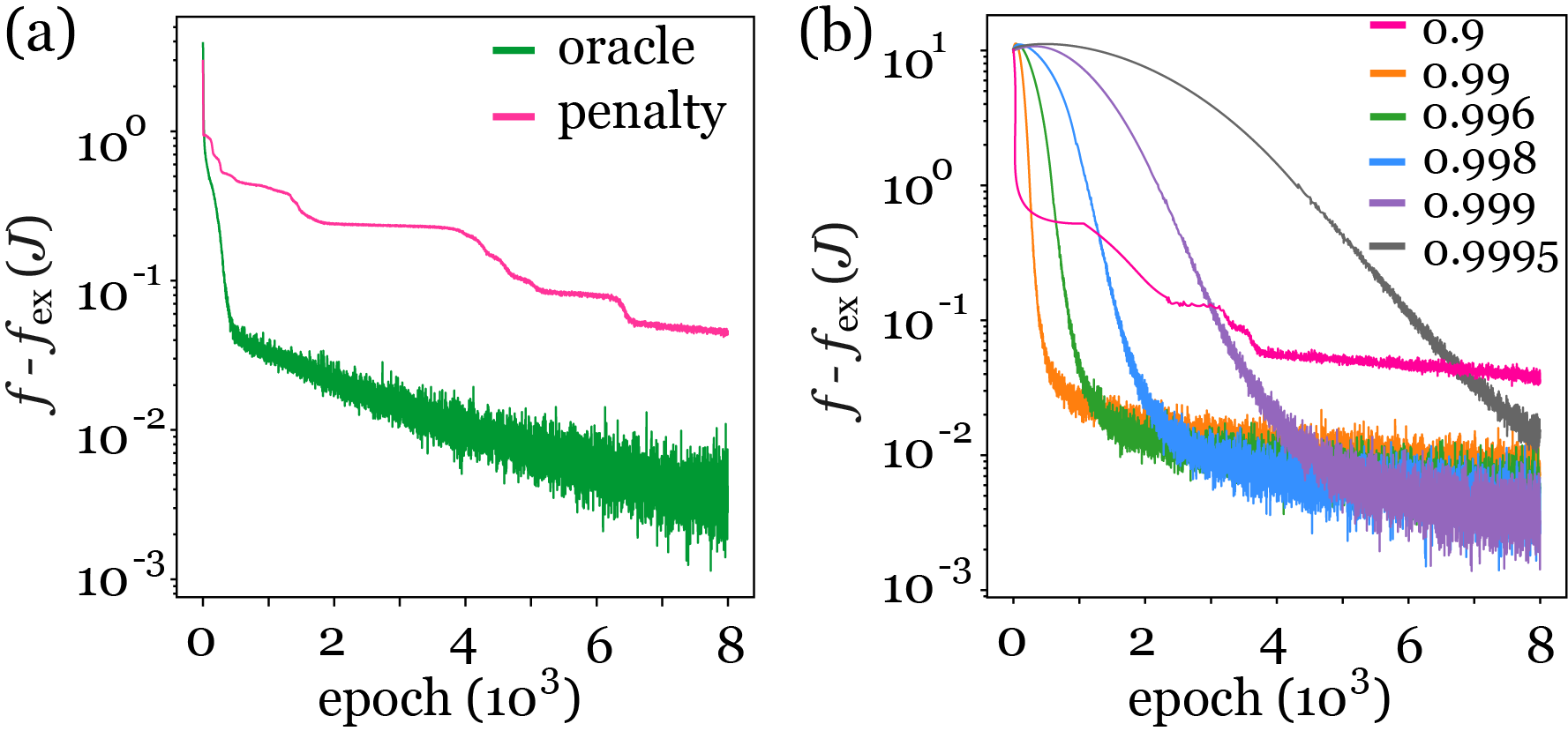}
	\caption{\label{fig:fig3}
	(a) The VAN solution of a 16-site Heisenberg chain, using oracle or energy penalty for eliminating dead ends. The difference between computed and exact free energies is plotted as a function of epoch. 
	(b) The oracle-guided VAN solution of a 32-site Heisenberg model, with different $r$ values for action rescaling.
	$\beta J=1$ and $M=8$ in both (a) and (b).
	}
\end{figure}

\section{Results}

Now we apply the oracle-guided fully visible sigmoidal belief network\citep{Frey98} to examine the thermodynamics of the quantum $XXZ$ spin chain in a few regimes, namely the antiferromagnetic ($\lambda=1$) and $XY$ ($\lambda= 1/8$) limits over a range of temperatures. In addition to the free energy as given in Eq. (\ref{eq:fe}), we also compute specific heat and magnetic susceptibility. The specific heat per spin is defined as $C_V= N^{-1}\overbar{(\partial E/\partial T)_N}$, where $E=\partial S/\partial \beta $ and over bar denotes thermal averaging by VAN. Therefore, the specific heat is computed as the following average 
\begin{equation}
	C_V=\frac{k_{\text B}\beta ^2}{N} \left(\overbar{E^2} - \overbar{E}^2- \overbar{C}\right),
\end{equation}
where $C = \partial^2S/\partial \beta^2$.
Analytical expressions for $E$ and $C$ can be easily derived and sampled using the VAN. The susceptibility  is computed by sampling the total magnetic moments, $m_z=N^{-1}\sum s_i^z$, as 
\begin{equation}
	\chi =-\left(\frac{\partial^2 f}{\partial B^2}\right)_{T,N}= \beta \left (\overbar{m_z^2}-\overbar{m_z}^2\right).
\end{equation}
The same quantities are also computed with stochastic series expansion (SSE)\citep{Sandvik99,Alet05} and exact diagonalization for calibrations.

\begin{figure}[ht]
	\includegraphics[width=85 mm, height=96.5 mm]{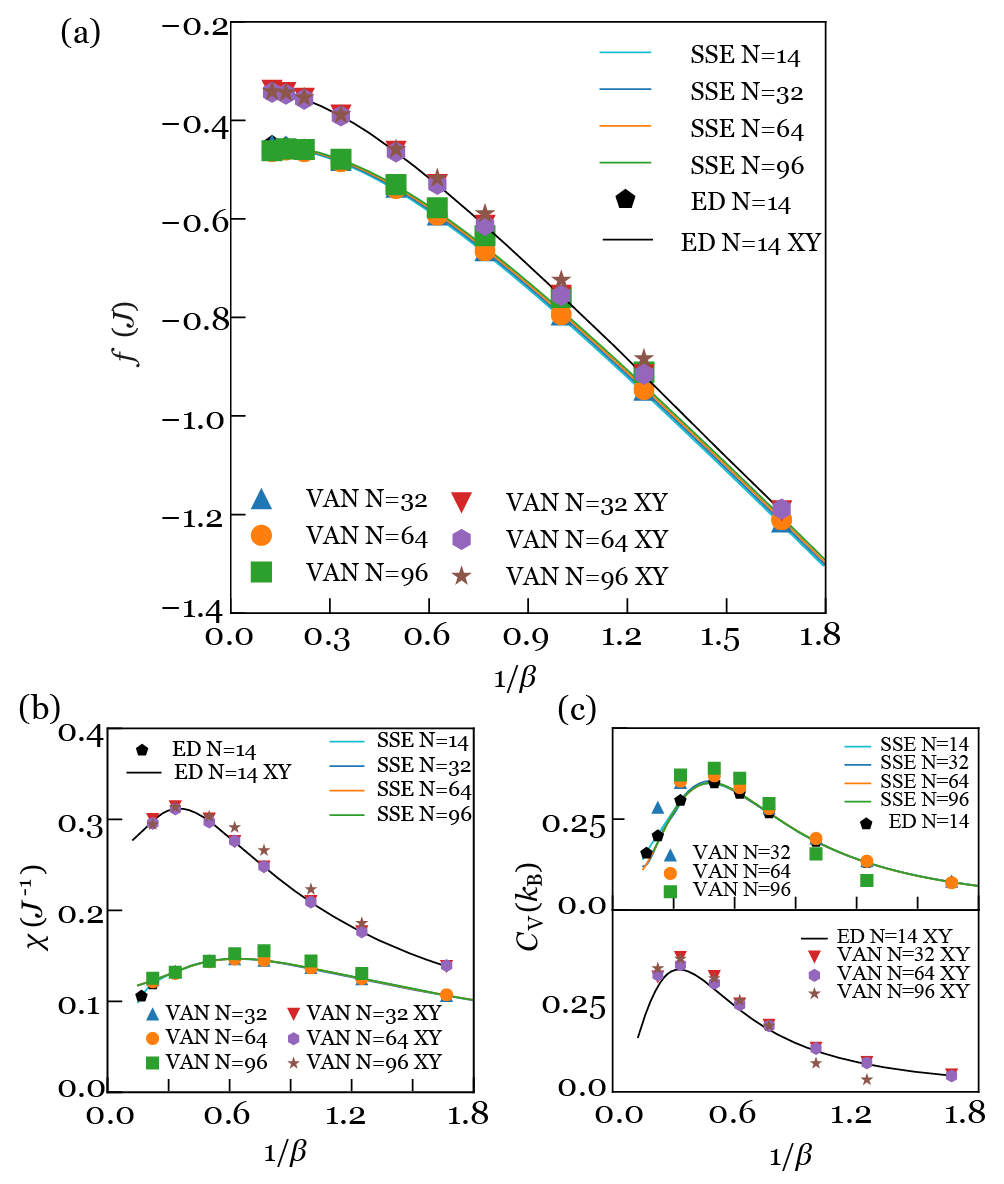}
	\caption{\label{fig:fig4} Various thermodynamic quantities computed using oracle-guided VAN, in comparison with SSE and exact diagonalization (ED) results.
	(a) Free energies of antiferromagnetic and $XY$-like (marked in legends) chains, with $M=8$ or 16 in VAN (see text). 
	(b) Susceptibilities and (c) specific heats of Heisenberg and $XY$ chains.
	}
\end{figure}

Shown in Fig. \ref{fig:fig4}(a) are the free energies of antiferromagnetic and $XY$-like chains with 32, 64 and 96 sites ($J>0$). The number of time slices is  $M=8$, except for $N=32$ and $\beta\ge 2$, where $M=16$ is used. We see that for the antiferromagnet, the VAN-computed  free energies agree with the SSE results reasonably well, only showing slight overestimation at elevated temperatures. It is also seen that the exact diagonalization results for shorter chains agree very well with SSE results regardless of chain lengths, showing little finite size effect. For the $XY$-like chains, we find that the free energies all fall very close to that from exact diagonalization of a 14-site chain. Since in this regime, the finite size effect is seen to be minimal, the computed free energies for  the $XY$ chains are also accurate. We would like to remark that even though VAN produces accurate free energies in the temperature range studied here, the performance for the antiferromagnetic Ising limit is not satisfactory, likely due to large Trotter discretization error, and very slow unbolting from localized modes exaggerating the occurrence of local excitations.

The computed susceptibilities shown in Fig. \ref{fig:fig4}(b) again  display little dependences on system sizes, and both show maxima at finite temperatures. For the antiferromagnetic case, where SSE solutions are available, the VAN-computed susceptibilities agree well with the Monte Carlo results and exact results. For the $XY$-like chains, the computed susceptibilities are considerably larger than the antiferromagnetic chains as expected, due to the floppiness of spins in this regime. In both cases, the susceptibilities for the long chain ($N=96$) show slightly larger deviations from references. The computed specific heats are shown in Fig. \ref{fig:fig4}(c), which are in general agreement with the reference values offered and show the maxima at finite temperatures, but have somewhat larger deviation than susceptibilities shown in Fig. \ref{fig:fig4}(b).  

\section{Conclusion}
In summary, 
we have constructed a variational autoregressive network, which efficiently and accurately represents the partition function of quantum spin systems, which is subsequently solved variationally. An oracle accompanying this VAN approach allows for rejection-free sampling satisfying the requisite boundary condition. Clearly, this approach to quantum statistical problem allows for highly parallelizable solutions of quantum many-body systems at finite temperatures. The oracle, based on a width-first greedy algorithm, results in significant reduction of computational cost compare to generating each spin by VAN, which can be applied to worldline generation in other methods, such as quantum Monte Carlo methods(worldline Monte Carlo,~\cite{Suzuki76, Suzuki77, Hirsch82, Evertz03} continuous time quantum Monte Carlo,\cite{Kornilovitch98, Rubtsov05} stochastic series expansion\cite{Sandvik99,Alet05}). An immediate future direction is to extend this approach to problems at higher dimensions, as well as quantum many-body problems of other types, such as Bose-Hubbard models, t-J models. Another example is the quantum dimer model, where the local constraints can be baffling to cluster-flipping quantum Monte Carlo, the VAN method may provide a viable alternative. A main limitation of VAN method is the fact that its sample generation is sequential, which is not efficient especially when a large number of spins is needed to represent a worldline configuration. A potential direction to explore is how graph neural network~\cite{Scarselli08} may offer a better representation of the configuration.

\begin{acknowledgments}
This work is supported by the National Natural Science Foundation of China (Grant No. 11725415 and No. 11934001), the Ministry of Science and Technology of China (Grant No. 2018YFA0305601 and No. 2016YFA0301004), and by the Strategic Priority Research Program of Chinese Academy of Sciences (Grant No. XDB28000000). 
\end{acknowledgments}


\appendix

\renewcommand{\thefigure}{A-\arabic{figure}}
\setcounter{figure}{0}

\renewcommand{\theequation}{A-\arabic{equation}}
\setcounter{equation}{0}


\section*{Appendix A. When to turn on oracle}

In this section, we provide the proof of the lemma used in the paper regarding the onset of unviable configurations when generating spin configuration on a spacetime lattice.  We will have to start with a  few definitions before a statement of the lemma is presented.  The proof for the lemma is comprised of two propositions.

\begin{dfn}
	A \underline{checkerboard} is a pattern formed by dividing a $N\times H$ rectangle  into $N\times H$ unit squares, which are then colored white and dark alternately in both directions. Periodic boundary condition is imposed in the horizontal direction only, so $N$ is even.
\end{dfn}

The direction along which there are $N$ squares is called horizontal direction, and the other is called vertical direction. In view of the periodic boundary condition, the checkerboard is a finite cylindrical surface and can be stereographically mapped onto an annulus, of which WLOG the inner circle is the bottom edge, and the outer circle the upper edge, as shown in Fig. \ref{fig:figA1}.

\begin{figure}[h]
	\centering
	\includegraphics[width=80
	mm]{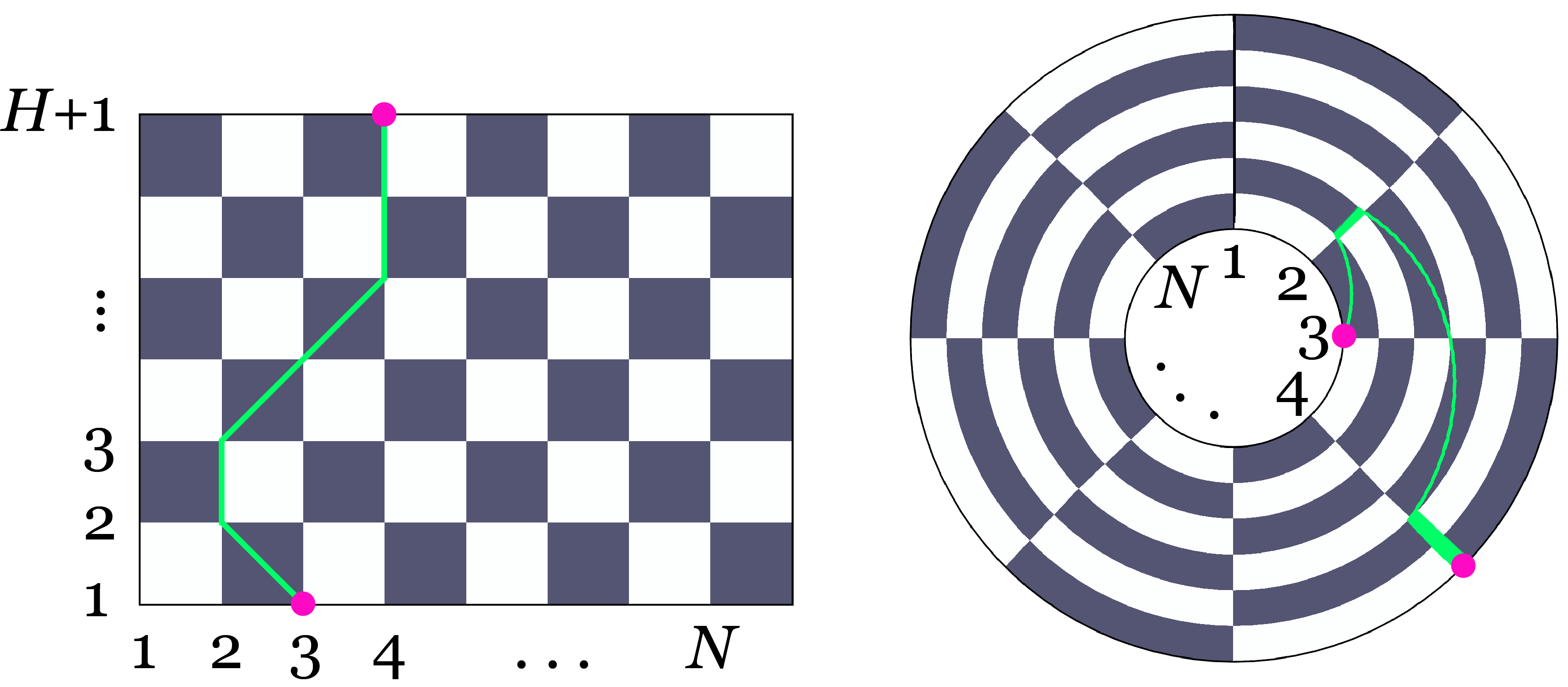} \\
	\caption{\label{fig:figA1}
	 A path on a $N\times H$ checkerboard.
	}
	
\end{figure}

\begin{dfn}
Corners of the $N\times H$ squares are called \underline{grid points}.
	 A $B\times H$ checkerboard has $N\times (H+1)$ grid points.
\end{dfn}

The locations of a grid point is specified by an ordered doublet, $(n, h)$, $n=1,\ldots,N, h=1,\ldots, H+1$. For the problem we will discuss, there are $k$ points located at distinct grid points, $\{x|x\in \mathbb Z, 1\le x\le N\}$, on the bottom edge, and $k$ points on the top edge at $\{y|y\in \mathbb Z, 1\le y\le N\}$.

\begin{dfn}
	A \underline{path} connects pair of grid points on the top and bottom edges of the checker board, and is as a collection of connected and directed line segments. Every line segment is upward directed starting and ending on grid points, and is either  vertical or a diagonal of a dark square.
\end{dfn}

When $H$ is small, a path may be be found for a pair of edge points. The question is: what is the lower bound of $H$ for a given $N$, a path can be found for a pair of edge points. What if there are more than one pair of points? The answer is Lemma 1:

\begin{lem}
	When $H\ge N/2+1$, a set of paths can always be found connecting $k \le N$ arbitrary grid points on the lower edge to $k $ arbitrary grid points on the top edge.
	\label{lem:01}
\end{lem}

The proof for Lemma 1 consists of two proposition/definitions.

\begin{propdfn}
	On each of two concentric circles lie $k$ points. N rays from the center can be drawn to partition the plane into N sections. For an arbitrary N, we can always find a \underline{partitioning} where the two arcs in each \underline{section} have identical numbers of points, and the angle made by abutting rays are less than or equal to 180~$^\circ$.
	\label{prop:01}
\end{propdfn}

Proposition \ref{prop:01} can be demonstrated by a constructive procedure for achieving this kind of partitioning. So long as $k\ge 1$, there must a point such that if we draw a ray through it and rotate the ray clockwise through an angle $\le180^\circ$, the first point the ray encounters lies on the other circle. We then draw rays through this pair of points to create a section, and erase this pair of points. Repeating the above process $k$ times, we will have $k$ sections each containing a pair of points located on the two circles. To obtain the final partitioning, we take unions of overlapping and neighboring sections with the constraint that the angle subtended by a final section is no greater than 180$^\circ$.

As a consequence of Proposition \ref{prop:01}, we can partition the checkerboard into non-empty sections, and in each section there are equal number of points on the top and bottom edges, and the horizontal separations between points are no greater than $N/2$.  The partitioning resulted from the above algorithm is usually not unique, and may not be the simplest; Proposition/Definition \ref{prop:01} merely asserts its existence. In $\sigma$th section, we number the points on the top and bottom edges separately, as $i_\sigma =1 \dots n_\sigma$ in increasing order from left to right. 

\begin{figure}[h]
	\centering
	\includegraphics[width=80mm]{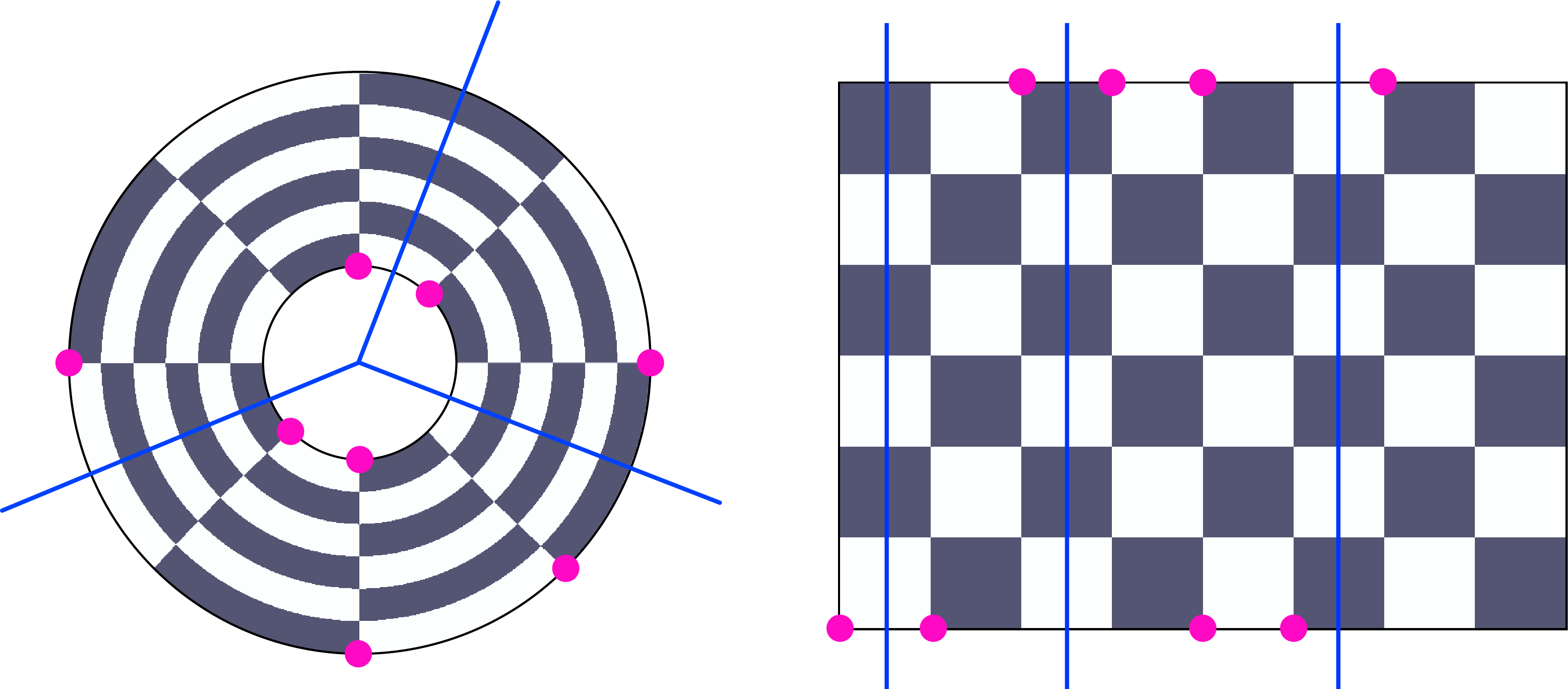} \\
	\caption{\label{fig:figA2}
	 Partition of unidirectional sections.
	}
\end{figure}

\begin{propdfn}
	In a \underline{unidirectional section} $\sigma$, $x(i_\sigma) \le y(i_\sigma)$ for all $i_\sigma$. If its width $+1\le H$,  then the points on the top and bottom edges can be pair-wise connected by nonintersecting paths.
	\label{prop:02}
\end{propdfn}

\begin{proof}
	We first show that in a \underline{contiguous  section} $\sigma$ where $\{x(i_\sigma)\}$  are consecutive integers and so are $\{y(i_\sigma)\}$ (which then is unidirectional), pairwise paths can be constructed. WLOG, we suppose $\ell_\sigma = y(1_\sigma) -x(1_\sigma)\ge 0$. 
	
	In the first step, a path connecting the last ($n_\sigma$th) pair of points is constructed by the procedure: starting at $(x(n_\sigma), 1)$ we draw a vertical line to the first grid point ($x(n_\sigma), h_\sigma$)  upper-right to whom a dark square lies, and draw along the diagonals to the grid point beneath $y(n_\sigma)$, and then draw a vertical line to reach the top edge. Since  $H\ge l_\sigma +1$, this path can always be drawn.
	
	Subsequently, similar procedure is performed consecutively for each pair to the left: a vertical line is drawn from $(x(n_\sigma -j), 1)$ to $(x(n_\sigma-j), h_\sigma+j)$, which then continues diagonally upward and to the right till beneath $y(n_\sigma-j)$ and then continues upward. The diagonal part of the path can be completed if
	\begin{eqnarray*}
		H-(h_\sigma+n_\sigma-1) \ge \ell_\sigma.
	\end{eqnarray*}
	Since $\ell_\sigma +n_\sigma  \le H$ and $h_\sigma = 0$ or 1, the inequality holds.
	
	For a unidirectional section that is not contiguous, paths can be constructed from those of an auxiliary contiguous section with  the same numbers of points on the top and bottom edges at $\{\tilde x(j_\sigma)\}$ and $\{\tilde y(j_\sigma)\}$, respectively, where $\tilde{x}(1)=x(1_\sigma), \tilde y(n_\sigma )=y(n_\sigma)$. Now starting at $(x(j_\sigma),1)$ a path will travel upward till it meets  auxiliary path $j_\sigma$ , then travels along with the path $j_\sigma$ until beneath $y(j_\sigma)$, and then continues vertically to $(y(j_\sigma),H+1)$. This is always possible since $\tilde x(j_\sigma) \le x(j_\sigma)$ and $\tilde y(j_\sigma) \ge y(j_\sigma).$
\end{proof}

Obviously, each non-unidirectional section can always be further divided into unidirectional sections. The proof for Lemma \ref{lem:01} is completed. Incidentally, since the paths constructed in our proof always travel  from the starting point on the lower edge toward the end point in the horizontal direction, we come to a corollary:
\begin{cor}
	If $H\ge N/2+1$, at least one greedy solution exists. 
\end{cor}

\section*{Appendix B. Exact diagonalization}

We exactly diagonalize the many-body Hamiltonian of short spin chains to provide a reference for the results from VAN. The many-body Hamiltonian is constructed on a basis of spin configurations $\ket{\{s^z\}}=\Pi_{j}\ket{s^z_{j}}$.  Rewriting Eq. (\ref{eq:xxz}) as
\begin{equation}
H = \sum_{j} \lambda s^{z}_j s^{z}_{j+1}+\frac{1}{2}(s^{+}_js^{-}_{j+1}+H.c.)
\end{equation}
where $s^{\pm}_j=s^x_j\pm is^y_j$, and $B=0$.

For example, the basis of a two-spin system can be chosen as $\ket{\downarrow\downarrow}, \ket{\uparrow\downarrow}, \ket{\downarrow\uparrow}, \ket{\uparrow\uparrow}$,and the hamiltonian matrix will be 
\begin{equation}
H_{N=2}=J\left(
\begin{array}{cccc}
\lambda/4 & 0 & 0 & 0 \\
0 & -\lambda/4 & 1/2 & 0 \\
0 & 1/2 & -\lambda/4 & 0 \\
0 & 0 & 0 & \lambda/4
\end{array}
\right)
\label{eq:hmatrix_two}
\end{equation}

We can see that in (\ref{eq:hmatrix_two}) the Hamiltonian matrix is block diagonal, which is a natural result from the conservation of total spins. In general, the N-spins hamiltonian has $N+1$ blocks $H_{N,n_{\uparrow}}$, where $n_{\uparrow}=0,1,\cdots,N$ is the number of up-spins. Computing the eigenvalues  $E^{(i)}_{N,n_{\uparrow}}$  of $H_{N,n_{\uparrow}}$ using \texttt{numpy} offered in standard python, where $i=1,2,\cdots,\frac{N!}{n_{\uparrow}!(N-n_{\uparrow})!}$,  we obtain the partition function as
\begin{equation}
Z = Tr[e^{-\beta H_N}]=\sum_{n_{\uparrow}}\sum_i e^{-\beta E^{(i)}_{N,n_{\uparrow}}}.
\end{equation}

%

%


\begin{thebibliography}{40}%
\makeatletter
\providecommand \@ifxundefined [1]{%
 \@ifx{#1\undefined}
}%
\providecommand \@ifnum [1]{%
 \ifnum #1\expandafter \@firstoftwo
 \else \expandafter \@secondoftwo
 \fi
}%
\providecommand \@ifx [1]{%
 \ifx #1\expandafter \@firstoftwo
 \else \expandafter \@secondoftwo
 \fi
}%
\providecommand \natexlab [1]{#1}%
\providecommand \enquote  [1]{``#1''}%
\providecommand \bibnamefont  [1]{#1}%
\providecommand \bibfnamefont [1]{#1}%
\providecommand \citenamefont [1]{#1}%
\providecommand \href@noop [0]{\@secondoftwo}%
\providecommand \href [0]{\begingroup \@sanitize@url \@href}%
\providecommand \@href[1]{\@@startlink{#1}\@@href}%
\providecommand \@@href[1]{\endgroup#1\@@endlink}%
\providecommand \@sanitize@url [0]{\catcode `\\12\catcode `\$12\catcode
  `\&12\catcode `\#12\catcode `\^12\catcode `\_12\catcode `\%12\relax}%
\providecommand \@@startlink[1]{}%
\providecommand \@@endlink[0]{}%
\providecommand \url  [0]{\begingroup\@sanitize@url \@url }%
\providecommand \@url [1]{\endgroup\@href {#1}{\urlprefix }}%
\providecommand \urlprefix  [0]{URL }%
\providecommand \Eprint [0]{\href }%
\providecommand \doibase [0]{http://dx.doi.org/}%
\providecommand \selectlanguage [0]{\@gobble}%
\providecommand \bibinfo  [0]{\@secondoftwo}%
\providecommand \bibfield  [0]{\@secondoftwo}%
\providecommand \translation [1]{[#1]}%
\providecommand \BibitemOpen [0]{}%
\providecommand \bibitemStop [0]{}%
\providecommand \bibitemNoStop [0]{.\EOS\space}%
\providecommand \EOS [0]{\spacefactor3000\relax}%
\providecommand \BibitemShut  [1]{\csname bibitem#1\endcsname}%
\let\auto@bib@innerbib\@empty
\bibitem [{\citenamefont {Feynman}(1953)}]{Feynman53}%
  \BibitemOpen
  \bibfield  {author} {\bibinfo {author} {\bibfnamefont {R.~P.}\ \bibnamefont
  {Feynman}},\ }\href {\doibase 10.1103/PhysRev.91.1291} {\bibfield  {journal}
  {\bibinfo  {journal} {Phys. Rev.}\ }\textbf {\bibinfo {volume} {91}},\
  \bibinfo {pages} {1291} (\bibinfo {year} {1953})}\BibitemShut {NoStop}%
\bibitem [{\citenamefont {Suzuki}(1976)}]{Suzuki76}%
  \BibitemOpen
  \bibfield  {author} {\bibinfo {author} {\bibfnamefont {M.}~\bibnamefont
  {Suzuki}},\ }\href {https://doi.org/10.1143/PTP.56.1454} {\bibfield
  {journal} {\bibinfo  {journal} {Prog. Theor. Phys.}\ }\textbf {\bibinfo
  {volume} {56}},\ \bibinfo {pages} {1454} (\bibinfo {year}
  {1976})}\BibitemShut {NoStop}%
\bibitem [{\citenamefont {Suzuki}\ \emph {et~al.}(1977)\citenamefont {Suzuki},
  \citenamefont {Miyashita},\ and\ \citenamefont {Kuroda}}]{Suzuki77}%
  \BibitemOpen
  \bibfield  {author} {\bibinfo {author} {\bibfnamefont {M.}~\bibnamefont
  {Suzuki}}, \bibinfo {author} {\bibfnamefont {S.}~\bibnamefont {Miyashita}}, \
  and\ \bibinfo {author} {\bibfnamefont {A.}~\bibnamefont {Kuroda}},\ }\href
  {https://doi.org/10.1143/PTP.58.1377} {\bibfield  {journal} {\bibinfo
  {journal} {Prog. Theor. Phys.}\ }\textbf {\bibinfo {volume} {58}},\ \bibinfo
  {pages} {1377} (\bibinfo {year} {1977})}\BibitemShut {NoStop}%
\bibitem [{\citenamefont {Hirsch}\ \emph {et~al.}(1982)\citenamefont {Hirsch},
  \citenamefont {Sugar}, \citenamefont {Scalapino},\ and\ \citenamefont
  {Blankenbecler}}]{Hirsch82}%
  \BibitemOpen
  \bibfield  {author} {\bibinfo {author} {\bibfnamefont {J.~E.}\ \bibnamefont
  {Hirsch}}, \bibinfo {author} {\bibfnamefont {R.~L.}\ \bibnamefont {Sugar}},
  \bibinfo {author} {\bibfnamefont {D.~J.}\ \bibnamefont {Scalapino}}, \ and\
  \bibinfo {author} {\bibfnamefont {R.}~\bibnamefont {Blankenbecler}},\ }\href
  {\doibase 10.1103/PhysRevB.26.5033} {\bibfield  {journal} {\bibinfo
  {journal} {Phys. Rev. B}\ }\textbf {\bibinfo {volume} {26}},\ \bibinfo
  {pages} {5033} (\bibinfo {year} {1982})}\BibitemShut {NoStop}%
\bibitem [{\citenamefont {Evertz}(2003)}]{Evertz03}%
  \BibitemOpen
  \bibfield  {author} {\bibinfo {author} {\bibfnamefont {H.~G.}\ \bibnamefont
  {Evertz}},\ }\href {\doibase 10.1080/0001873021000049195} {\bibfield
  {journal} {\bibinfo  {journal} {Adv. Phys.}\ }\textbf {\bibinfo {volume}
  {52}},\ \bibinfo {pages} {1} (\bibinfo {year} {2003})}\BibitemShut {NoStop}%
\bibitem [{\citenamefont {Carleo}\ and\ \citenamefont
  {Troyer}(2017)}]{Carleo602}%
  \BibitemOpen
  \bibfield  {author} {\bibinfo {author} {\bibfnamefont {G.}~\bibnamefont
  {Carleo}}\ and\ \bibinfo {author} {\bibfnamefont {M.}~\bibnamefont
  {Troyer}},\ }\href {\doibase 10.1126/science.aag2302} {\bibfield  {journal}
  {\bibinfo  {journal} {Science}\ }\textbf {\bibinfo {volume} {355}},\ \bibinfo
  {pages} {602} (\bibinfo {year} {2017})}\BibitemShut {NoStop}%
\bibitem [{\citenamefont {Saito}\ and\ \citenamefont {Kato}(2018)}]{Saito17}%
  \BibitemOpen
  \bibfield  {author} {\bibinfo {author} {\bibfnamefont {H.}~\bibnamefont
  {Saito}}\ and\ \bibinfo {author} {\bibfnamefont {M.}~\bibnamefont {Kato}},\
  }\href {\doibase 10.7566/JPSJ.87.014001} {\bibfield  {journal} {\bibinfo
  {journal} {Journal of the Physical Society of Japan}\ }\textbf {\bibinfo
  {volume} {87}},\ \bibinfo {pages} {014001} (\bibinfo {year}
  {2018})}\BibitemShut {NoStop}%
\bibitem [{\citenamefont {Deng}\ \emph {et~al.}(2017)\citenamefont {Deng},
  \citenamefont {Li},\ and\ \citenamefont {Das~Sarma}}]{Deng17}%
  \BibitemOpen
  \bibfield  {author} {\bibinfo {author} {\bibfnamefont {D.-L.}\ \bibnamefont
  {Deng}}, \bibinfo {author} {\bibfnamefont {X.}~\bibnamefont {Li}}, \ and\
  \bibinfo {author} {\bibfnamefont {S.}~\bibnamefont {Das~Sarma}},\ }\href
  {\doibase 10.1103/PhysRevB.96.195145} {\bibfield  {journal} {\bibinfo
  {journal} {Phys. Rev. B}\ }\textbf {\bibinfo {volume} {96}},\ \bibinfo
  {pages} {195145} (\bibinfo {year} {2017})}\BibitemShut {NoStop}%
\bibitem [{\citenamefont {Glasser}\ \emph {et~al.}(2018)\citenamefont
  {Glasser}, \citenamefont {Pancotti}, \citenamefont {August}, \citenamefont
  {Rodriguez},\ and\ \citenamefont {Cirac}}]{Glasser18}%
  \BibitemOpen
  \bibfield  {author} {\bibinfo {author} {\bibfnamefont {I.}~\bibnamefont
  {Glasser}}, \bibinfo {author} {\bibfnamefont {N.}~\bibnamefont {Pancotti}},
  \bibinfo {author} {\bibfnamefont {M.}~\bibnamefont {August}}, \bibinfo
  {author} {\bibfnamefont {I.~D.}\ \bibnamefont {Rodriguez}}, \ and\ \bibinfo
  {author} {\bibfnamefont {J.~I.}\ \bibnamefont {Cirac}},\ }\href {\doibase
  10.1103/PhysRevX.8.011006} {\bibfield  {journal} {\bibinfo  {journal} {Phys.
  Rev. X}\ }\textbf {\bibinfo {volume} {8}},\ \bibinfo {pages} {011006}
  (\bibinfo {year} {2018})}\BibitemShut {NoStop}%
\bibitem [{\citenamefont {Vieijra}\ \emph {et~al.}(2020)\citenamefont
  {Vieijra}, \citenamefont {Casert}, \citenamefont {Nys}, \citenamefont
  {De~Neve}, \citenamefont {Haegeman}, \citenamefont {Ryckebusch},\ and\
  \citenamefont {Verstraete}}]{Vieijra19}%
  \BibitemOpen
  \bibfield  {author} {\bibinfo {author} {\bibfnamefont {T.}~\bibnamefont
  {Vieijra}}, \bibinfo {author} {\bibfnamefont {C.}~\bibnamefont {Casert}},
  \bibinfo {author} {\bibfnamefont {J.}~\bibnamefont {Nys}}, \bibinfo {author}
  {\bibfnamefont {W.}~\bibnamefont {De~Neve}}, \bibinfo {author} {\bibfnamefont
  {J.}~\bibnamefont {Haegeman}}, \bibinfo {author} {\bibfnamefont
  {J.}~\bibnamefont {Ryckebusch}}, \ and\ \bibinfo {author} {\bibfnamefont
  {F.}~\bibnamefont {Verstraete}},\ }\href {\doibase
  10.1103/PhysRevLett.124.097201} {\bibfield  {journal} {\bibinfo  {journal}
  {Phys. Rev. Lett.}\ }\textbf {\bibinfo {volume} {124}},\ \bibinfo {pages}
  {097201} (\bibinfo {year} {2020})}\BibitemShut {NoStop}%
\bibitem [{\citenamefont {Zheng}\ \emph {et~al.}(2019)\citenamefont {Zheng},
  \citenamefont {He}, \citenamefont {Regnault},\ and\ \citenamefont
  {Bernevig}}]{Zheng19}%
  \BibitemOpen
  \bibfield  {author} {\bibinfo {author} {\bibfnamefont {Y.}~\bibnamefont
  {Zheng}}, \bibinfo {author} {\bibfnamefont {H.}~\bibnamefont {He}}, \bibinfo
  {author} {\bibfnamefont {N.}~\bibnamefont {Regnault}}, \ and\ \bibinfo
  {author} {\bibfnamefont {B.~A.}\ \bibnamefont {Bernevig}},\ }\href {\doibase
  10.1103/PhysRevB.99.155129} {\bibfield  {journal} {\bibinfo  {journal} {Phys.
  Rev. B}\ }\textbf {\bibinfo {volume} {99}},\ \bibinfo {pages} {155129}
  (\bibinfo {year} {2019})}\BibitemShut {NoStop}%
\bibitem [{\citenamefont {Choo}\ \emph {et~al.}(2018)\citenamefont {Choo},
  \citenamefont {Carleo}, \citenamefont {Regnault},\ and\ \citenamefont
  {Neupert}}]{Choo18}%
  \BibitemOpen
  \bibfield  {author} {\bibinfo {author} {\bibfnamefont {K.}~\bibnamefont
  {Choo}}, \bibinfo {author} {\bibfnamefont {G.}~\bibnamefont {Carleo}},
  \bibinfo {author} {\bibfnamefont {N.}~\bibnamefont {Regnault}}, \ and\
  \bibinfo {author} {\bibfnamefont {T.}~\bibnamefont {Neupert}},\ }\href
  {\doibase 10.1103/PhysRevLett.121.167204} {\bibfield  {journal} {\bibinfo
  {journal} {Phys. Rev. Lett.}\ }\textbf {\bibinfo {volume} {121}},\ \bibinfo
  {pages} {167204} (\bibinfo {year} {2018})}\BibitemShut {NoStop}%
\bibitem [{\citenamefont {Yoshioka}\ and\ \citenamefont
  {Hamazaki}(2019)}]{Yoshioka19}%
  \BibitemOpen
  \bibfield  {author} {\bibinfo {author} {\bibfnamefont {N.}~\bibnamefont
  {Yoshioka}}\ and\ \bibinfo {author} {\bibfnamefont {R.}~\bibnamefont
  {Hamazaki}},\ }\href {\doibase 10.1103/PhysRevB.99.214306} {\bibfield
  {journal} {\bibinfo  {journal} {Phys. Rev. B}\ }\textbf {\bibinfo {volume}
  {99}},\ \bibinfo {pages} {214306} (\bibinfo {year} {2019})}\BibitemShut
  {NoStop}%
\bibitem [{\citenamefont {Nagy}\ and\ \citenamefont {Savona}(2019)}]{Nagy19}%
  \BibitemOpen
  \bibfield  {author} {\bibinfo {author} {\bibfnamefont {A.}~\bibnamefont
  {Nagy}}\ and\ \bibinfo {author} {\bibfnamefont {V.}~\bibnamefont {Savona}},\
  }\href {\doibase 10.1103/PhysRevLett.122.250501} {\bibfield  {journal}
  {\bibinfo  {journal} {Phys. Rev. Lett.}\ }\textbf {\bibinfo {volume} {122}},\
  \bibinfo {pages} {250501} (\bibinfo {year} {2019})}\BibitemShut {NoStop}%
\bibitem [{\citenamefont {Hartmann}\ and\ \citenamefont
  {Carleo}(2019)}]{Hartmann19}%
  \BibitemOpen
  \bibfield  {author} {\bibinfo {author} {\bibfnamefont {M.~J.}\ \bibnamefont
  {Hartmann}}\ and\ \bibinfo {author} {\bibfnamefont {G.}~\bibnamefont
  {Carleo}},\ }\href {\doibase 10.1103/PhysRevLett.122.250502} {\bibfield
  {journal} {\bibinfo  {journal} {Phys. Rev. Lett.}\ }\textbf {\bibinfo
  {volume} {122}},\ \bibinfo {pages} {250502} (\bibinfo {year}
  {2019})}\BibitemShut {NoStop}%
\bibitem [{\citenamefont {Vicentini}\ \emph {et~al.}(2019)\citenamefont
  {Vicentini}, \citenamefont {Biella}, \citenamefont {Regnault},\ and\
  \citenamefont {Ciuti}}]{Vicentini19}%
  \BibitemOpen
  \bibfield  {author} {\bibinfo {author} {\bibfnamefont {F.}~\bibnamefont
  {Vicentini}}, \bibinfo {author} {\bibfnamefont {A.}~\bibnamefont {Biella}},
  \bibinfo {author} {\bibfnamefont {N.}~\bibnamefont {Regnault}}, \ and\
  \bibinfo {author} {\bibfnamefont {C.}~\bibnamefont {Ciuti}},\ }\href
  {\doibase 10.1103/PhysRevLett.122.250503} {\bibfield  {journal} {\bibinfo
  {journal} {Phys. Rev. Lett.}\ }\textbf {\bibinfo {volume} {122}},\ \bibinfo
  {pages} {250503} (\bibinfo {year} {2019})}\BibitemShut {NoStop}%
\bibitem [{\citenamefont {Wu}\ \emph {et~al.}(2019)\citenamefont {Wu},
  \citenamefont {Wang},\ and\ \citenamefont {Zhang}}]{Wu19}%
  \BibitemOpen
  \bibfield  {author} {\bibinfo {author} {\bibfnamefont {D.}~\bibnamefont
  {Wu}}, \bibinfo {author} {\bibfnamefont {L.}~\bibnamefont {Wang}}, \ and\
  \bibinfo {author} {\bibfnamefont {P.}~\bibnamefont {Zhang}},\ }\href
  {\doibase 10.1103/PhysRevLett.122.080602} {\bibfield  {journal} {\bibinfo
  {journal} {Phys. Rev. Lett.}\ }\textbf {\bibinfo {volume} {122}},\ \bibinfo
  {pages} {080602} (\bibinfo {year} {2019})}\BibitemShut {NoStop}%
\bibitem [{\citenamefont {Sharir}\ \emph {et~al.}(2020)\citenamefont {Sharir},
  \citenamefont {Levine}, \citenamefont {Wies}, \citenamefont {Carleo},\ and\
  \citenamefont {Shashua}}]{Sharir20}%
  \BibitemOpen
  \bibfield  {author} {\bibinfo {author} {\bibfnamefont {O.}~\bibnamefont
  {Sharir}}, \bibinfo {author} {\bibfnamefont {Y.}~\bibnamefont {Levine}},
  \bibinfo {author} {\bibfnamefont {N.}~\bibnamefont {Wies}}, \bibinfo {author}
  {\bibfnamefont {G.}~\bibnamefont {Carleo}}, \ and\ \bibinfo {author}
  {\bibfnamefont {A.}~\bibnamefont {Shashua}},\ }\href {\doibase
  10.1103/PhysRevLett.124.020503} {\bibfield  {journal} {\bibinfo  {journal}
  {Phys. Rev. Lett.}\ }\textbf {\bibinfo {volume} {124}},\ \bibinfo {pages}
  {020503} (\bibinfo {year} {2020})}\BibitemShut {NoStop}%
\bibitem [{\citenamefont {Luo}\ and\ \citenamefont {Clark}(2019)}]{Luo19}%
  \BibitemOpen
  \bibfield  {author} {\bibinfo {author} {\bibfnamefont {D.}~\bibnamefont
  {Luo}}\ and\ \bibinfo {author} {\bibfnamefont {B.~K.}\ \bibnamefont
  {Clark}},\ }\href {\doibase 10.1103/PhysRevLett.122.226401} {\bibfield
  {journal} {\bibinfo  {journal} {Phys. Rev. Lett.}\ }\textbf {\bibinfo
  {volume} {122}},\ \bibinfo {pages} {226401} (\bibinfo {year}
  {2019})}\BibitemShut {NoStop}%
\bibitem [{\citenamefont {Hermann}\ \emph {et~al.}(2020)\citenamefont
  {Hermann}, \citenamefont {Schätzle},\ and\ \citenamefont
  {Noé}}]{Hermann20}%
  \BibitemOpen
  \bibfield  {author} {\bibinfo {author} {\bibfnamefont {J.}~\bibnamefont
  {Hermann}}, \bibinfo {author} {\bibfnamefont {Z.}~\bibnamefont {Schätzle}},
  \ and\ \bibinfo {author} {\bibfnamefont {F.}~\bibnamefont {Noé}},\ }\href
  {\doibase 10.1038/s41557-020-0544-y} {\bibfield  {journal} {\bibinfo
  {journal} {Nat. Chem.}\ }\textbf {\bibinfo {volume} {12}},\ \bibinfo {pages}
  {891–897} (\bibinfo {year} {2020})}\BibitemShut {NoStop}%
\bibitem [{\citenamefont {Choo}\ \emph {et~al.}(2019)\citenamefont {Choo},
  \citenamefont {Neupert},\ and\ \citenamefont {Carleo}}]{Choo19}%
  \BibitemOpen
  \bibfield  {author} {\bibinfo {author} {\bibfnamefont {K.}~\bibnamefont
  {Choo}}, \bibinfo {author} {\bibfnamefont {T.}~\bibnamefont {Neupert}}, \
  and\ \bibinfo {author} {\bibfnamefont {G.}~\bibnamefont {Carleo}},\ }\href
  {\doibase 10.1103/PhysRevB.100.125124} {\bibfield  {journal} {\bibinfo
  {journal} {Phys. Rev. B}\ }\textbf {\bibinfo {volume} {100}},\ \bibinfo
  {pages} {125124} (\bibinfo {year} {2019})}\BibitemShut {NoStop}%
\bibitem [{\citenamefont {Carrasquilla}\ \emph
  {et~al.}(2019{\natexlab{a}})\citenamefont {Carrasquilla}, \citenamefont
  {Luo}, \citenamefont {Pérez}, \citenamefont {Milsted}, \citenamefont
  {Clark}, \citenamefont {Volkovs},\ and\ \citenamefont
  {Aolita}}]{carrasquilla2019}%
  \BibitemOpen
  \bibfield  {author} {\bibinfo {author} {\bibfnamefont {J.}~\bibnamefont
  {Carrasquilla}}, \bibinfo {author} {\bibfnamefont {D.}~\bibnamefont {Luo}},
  \bibinfo {author} {\bibfnamefont {F.}~\bibnamefont {Pérez}}, \bibinfo
  {author} {\bibfnamefont {A.}~\bibnamefont {Milsted}}, \bibinfo {author}
  {\bibfnamefont {B.~K.}\ \bibnamefont {Clark}}, \bibinfo {author}
  {\bibfnamefont {M.}~\bibnamefont {Volkovs}}, \ and\ \bibinfo {author}
  {\bibfnamefont {L.}~\bibnamefont {Aolita}},\ }\href@noop {} {} (\bibinfo
  {year} {2019}{\natexlab{a}}),\ \Eprint {http://arxiv.org/abs/1912.11052}
  {arXiv:1912.11052 [cond-mat.str-el]} \BibitemShut {NoStop}%
\bibitem [{\citenamefont {Hibat-Allah}\ \emph {et~al.}(2020)\citenamefont
  {Hibat-Allah}, \citenamefont {Ganahl}, \citenamefont {Hayward}, \citenamefont
  {Melko},\ and\ \citenamefont {Carrasquilla}}]{Hibat-Allah20}%
  \BibitemOpen
  \bibfield  {author} {\bibinfo {author} {\bibfnamefont {M.}~\bibnamefont
  {Hibat-Allah}}, \bibinfo {author} {\bibfnamefont {M.}~\bibnamefont {Ganahl}},
  \bibinfo {author} {\bibfnamefont {L.~E.}\ \bibnamefont {Hayward}}, \bibinfo
  {author} {\bibfnamefont {R.~G.}\ \bibnamefont {Melko}}, \ and\ \bibinfo
  {author} {\bibfnamefont {J.}~\bibnamefont {Carrasquilla}},\ }\href {\doibase
  10.1103/PhysRevResearch.2.023358} {\bibfield  {journal} {\bibinfo  {journal}
  {Phys. Rev. Research}\ }\textbf {\bibinfo {volume} {2}},\ \bibinfo {pages}
  {023358} (\bibinfo {year} {2020})}\BibitemShut {NoStop}%
\bibitem [{\citenamefont {Morawetz}\ \emph {et~al.}(2021)\citenamefont
  {Morawetz}, \citenamefont {De~Vlugt}, \citenamefont {Carrasquilla},\ and\
  \citenamefont {Melko}}]{Morawetz21}%
  \BibitemOpen
  \bibfield  {author} {\bibinfo {author} {\bibfnamefont {S.}~\bibnamefont
  {Morawetz}}, \bibinfo {author} {\bibfnamefont {I.~J.~S.}\ \bibnamefont
  {De~Vlugt}}, \bibinfo {author} {\bibfnamefont {J.}~\bibnamefont
  {Carrasquilla}}, \ and\ \bibinfo {author} {\bibfnamefont {R.~G.}\
  \bibnamefont {Melko}},\ }\href {\doibase 10.1103/PhysRevA.104.012401}
  {\bibfield  {journal} {\bibinfo  {journal} {Phys. Rev. A}\ }\textbf {\bibinfo
  {volume} {104}},\ \bibinfo {pages} {012401} (\bibinfo {year}
  {2021})}\BibitemShut {NoStop}%
\bibitem [{\citenamefont {Luo}\ \emph {et~al.}()\citenamefont {Luo},
  \citenamefont {Chen}, \citenamefont {Carrasquilla},\ and\ \citenamefont
  {Clark}}]{Luo20}%
  \BibitemOpen
  \bibfield  {author} {\bibinfo {author} {\bibfnamefont {D.}~\bibnamefont
  {Luo}}, \bibinfo {author} {\bibfnamefont {Z.}~\bibnamefont {Chen}}, \bibinfo
  {author} {\bibfnamefont {J.}~\bibnamefont {Carrasquilla}}, \ and\ \bibinfo
  {author} {\bibfnamefont {B.~K.}\ \bibnamefont {Clark}},\ }\href
  {https://arxiv.org/abs/2009.05580} {\ }\Eprint
  {http://arxiv.org/abs/2009.05580} {arXiv:2009.05580 [quant-ph]} \BibitemShut
  {NoStop}%
\bibitem [{\citenamefont {Carrasquilla}\ \emph
  {et~al.}(2019{\natexlab{b}})\citenamefont {Carrasquilla}, \citenamefont
  {Torlai},\ and\ \citenamefont {Melko}}]{Carrasquilla19}%
  \BibitemOpen
  \bibfield  {author} {\bibinfo {author} {\bibfnamefont {J.}~\bibnamefont
  {Carrasquilla}}, \bibinfo {author} {\bibfnamefont {G.}~\bibnamefont
  {Torlai}}, \ and\ \bibinfo {author} {\bibfnamefont {R.~G.}\ \bibnamefont
  {Melko}},\ }\href {\doibase 10.1038/s42256-019-0028-1} {\bibfield  {journal}
  {\bibinfo  {journal} {Nat Mach Intell}\ }\textbf {\bibinfo {volume} {1}},\
  \bibinfo {pages} {155} (\bibinfo {year} {2019}{\natexlab{b}})}\BibitemShut
  {NoStop}%
\bibitem [{\citenamefont {Liu}\ \emph {et~al.}(2021)\citenamefont {Liu},
  \citenamefont {Mao}, \citenamefont {Zhang},\ and\ \citenamefont
  {Wang}}]{Liu21}%
  \BibitemOpen
  \bibfield  {author} {\bibinfo {author} {\bibfnamefont {J.-G.}\ \bibnamefont
  {Liu}}, \bibinfo {author} {\bibfnamefont {L.}~\bibnamefont {Mao}}, \bibinfo
  {author} {\bibfnamefont {P.}~\bibnamefont {Zhang}}, \ and\ \bibinfo {author}
  {\bibfnamefont {L.}~\bibnamefont {Wang}},\ }\href {\doibase
  10.1088/2632-2153/aba19d} {\bibfield  {journal} {\bibinfo  {journal} {Mach.
  Learn.: Sci. Technol.}\ }\textbf {\bibinfo {volume} {2}},\ \bibinfo {pages}
  {025011} (\bibinfo {year} {2021})}\BibitemShut {NoStop}%
\bibitem [{\citenamefont {Frey}(1998)}]{Frey98}%
  \BibitemOpen
  \bibfield  {author} {\bibinfo {author} {\bibfnamefont {B.~J.}\ \bibnamefont
  {Frey}},\ }\href@noop {} {\bibfield  {journal} {\bibinfo  {journal}
  {Graphical models for machine learning and digital communication}\ }
  (\bibinfo {year} {MIT press, Cambridge, 1998})}\BibitemShut {NoStop}%
\bibitem [{\citenamefont {Goodfellow}\ \emph {et~al.}(2016)\citenamefont
  {Goodfellow}, \citenamefont {Bengio},\ and\ \citenamefont
  {Courville}}]{Goodfellow}%
  \BibitemOpen
  \bibfield  {author} {\bibinfo {author} {\bibfnamefont {I.}~\bibnamefont
  {Goodfellow}}, \bibinfo {author} {\bibfnamefont {Y.}~\bibnamefont {Bengio}},
  \ and\ \bibinfo {author} {\bibfnamefont {A.}~\bibnamefont {Courville}},\
  }\href@noop {} {\emph {\bibinfo {title} {Deep Learning}}}\ (\bibinfo
  {publisher} {MIT Press, Cambridge},\ \bibinfo {year} {2016})\BibitemShut
  {NoStop}%
\bibitem [{Notei()}]{Notei}%
  \BibitemOpen
  \bibinfo {note} {$q(s_{00}=\uparrow )=1/2$ if $Z_2$ symmetry is presented.
  otherwise $q(s_{00}=\uparrow )$ is determined by the network}\BibitemShut
  {NoStop}%
\bibitem [{\citenamefont {Goodfellow}\ \emph {et~al.}(2014)\citenamefont
  {Goodfellow}, \citenamefont {Pouget-Abadie}, \citenamefont {Mirza},
  \citenamefont {Xu}, \citenamefont {Warde-Farley}, \citenamefont {Ozair},
  \citenamefont {Courville},\ and\ \citenamefont {Bengio}}]{Goodfellow14}%
  \BibitemOpen
  \bibfield  {author} {\bibinfo {author} {\bibfnamefont {I.}~\bibnamefont
  {Goodfellow}}, \bibinfo {author} {\bibfnamefont {J.}~\bibnamefont
  {Pouget-Abadie}}, \bibinfo {author} {\bibfnamefont {M.}~\bibnamefont
  {Mirza}}, \bibinfo {author} {\bibfnamefont {B.}~\bibnamefont {Xu}}, \bibinfo
  {author} {\bibfnamefont {D.}~\bibnamefont {Warde-Farley}}, \bibinfo {author}
  {\bibfnamefont {S.}~\bibnamefont {Ozair}}, \bibinfo {author} {\bibfnamefont
  {A.}~\bibnamefont {Courville}}, \ and\ \bibinfo {author} {\bibfnamefont
  {Y.}~\bibnamefont {Bengio}},\ }\href
  {https://www.scopus.com/inward/record.uri?eid=2-s2.0-84937849144&partnerID=40&md5=6441b2a288c5fdded2adbcc8b21e092c}
  {\bibfield  {journal} {\bibinfo  {journal} {Adv. Neural Inform. Process.
  Syst.}\ }\textbf {\bibinfo {volume} {3}},\ \bibinfo {pages} {2672} (\bibinfo
  {year} {2014})}\BibitemShut {NoStop}%
\bibitem [{\citenamefont {Uria}\ \emph {et~al.}(2016)\citenamefont {Uria},
  \citenamefont {Côté}, \citenamefont {Gregor}, \citenamefont {Murray},\ and\
  \citenamefont {Larochelle}}]{Uria16}%
  \BibitemOpen
  \bibfield  {author} {\bibinfo {author} {\bibfnamefont {B.}~\bibnamefont
  {Uria}}, \bibinfo {author} {\bibfnamefont {M.-A.}\ \bibnamefont {Côté}},
  \bibinfo {author} {\bibfnamefont {K.}~\bibnamefont {Gregor}}, \bibinfo
  {author} {\bibfnamefont {I.}~\bibnamefont {Murray}}, \ and\ \bibinfo {author}
  {\bibfnamefont {H.}~\bibnamefont {Larochelle}},\ }\href@noop {} {\enquote
  {\bibinfo {title} {Neural autoregressive distribution estimation},}\ }
  (\bibinfo {year} {2016}),\ \Eprint {http://arxiv.org/abs/1605.02226}
  {arXiv:1605.02226 [cs.LG]} \BibitemShut {NoStop}%
\bibitem [{Noteii()}]{Noteii}%
  \BibitemOpen
  \bibinfo {note} {In context of computational complexity theory, an oracle is
  a black box attached to a Turing machine, which solves a certain decision
  problem.}\BibitemShut {Stop}%
\bibitem [{\citenamefont {Cormen}\ \emph {et~al.}(2009)\citenamefont {Cormen},
  \citenamefont {Leiserson}, \citenamefont {Rivest},\ and\ \citenamefont
  {Stein}}]{Cormen01}%
  \BibitemOpen
  \bibfield  {author} {\bibinfo {author} {\bibfnamefont {T.}~\bibnamefont
  {Cormen}}, \bibinfo {author} {\bibfnamefont {C.~E.}\ \bibnamefont
  {Leiserson}}, \bibinfo {author} {\bibfnamefont {R.~L.}\ \bibnamefont
  {Rivest}}, \ and\ \bibinfo {author} {\bibfnamefont {C.}~\bibnamefont
  {Stein}},\ }\href@noop {} {\emph {\bibinfo {title} {Introduction to
  Algorithms}}}\ (\bibinfo  {publisher} {MIT press, Cambridge},\ \bibinfo
  {year} {2009})\BibitemShut {NoStop}%
\bibitem [{\citenamefont {Che}\ \emph {et~al.}(2017)\citenamefont {Che},
  \citenamefont {Jacob}, \citenamefont {Bengio},\ and\ \citenamefont
  {Li}}]{Che17}%
  \BibitemOpen
  \bibfield  {author} {\bibinfo {author} {\bibfnamefont {T.}~\bibnamefont
  {Che}}, \bibinfo {author} {\bibfnamefont {Y.~L.~A.}\ \bibnamefont {Jacob}},
  \bibinfo {author} {\bibfnamefont {Y.}~\bibnamefont {Bengio}}, \ and\ \bibinfo
  {author} {\bibfnamefont {W.}~\bibnamefont {Li}},\ }in\ \href
  {https://openreview.net/pdf?id=HJKkY35le} {\emph {\bibinfo {booktitle}
  {{ICLR}}}}\ (\bibinfo {year} {2017})\BibitemShut {NoStop}%
\bibitem [{\citenamefont {Sandvik}(1999)}]{Sandvik99}%
  \BibitemOpen
  \bibfield  {author} {\bibinfo {author} {\bibfnamefont {A.~W.}\ \bibnamefont
  {Sandvik}},\ }\href {\doibase 10.1103/PhysRevB.59.R14157} {\bibfield
  {journal} {\bibinfo  {journal} {Phys. Rev. B}\ }\textbf {\bibinfo {volume}
  {59}},\ \bibinfo {pages} {R14157} (\bibinfo {year} {1999})}\BibitemShut
  {NoStop}%
\bibitem [{\citenamefont {Alet}\ \emph {et~al.}(2005)\citenamefont {Alet},
  \citenamefont {Wessel},\ and\ \citenamefont {Troyer}}]{Alet05}%
  \BibitemOpen
  \bibfield  {author} {\bibinfo {author} {\bibfnamefont {F.}~\bibnamefont
  {Alet}}, \bibinfo {author} {\bibfnamefont {S.}~\bibnamefont {Wessel}}, \ and\
  \bibinfo {author} {\bibfnamefont {M.}~\bibnamefont {Troyer}},\ }\href
  {\doibase 10.1103/PhysRevE.71.036706} {\bibfield  {journal} {\bibinfo
  {journal} {Phys. Rev. E}\ }\textbf {\bibinfo {volume} {71}},\ \bibinfo
  {pages} {036706} (\bibinfo {year} {2005})}\BibitemShut {NoStop}%
\bibitem [{\citenamefont {Kornilovitch}(1998)}]{Kornilovitch98}%
  \BibitemOpen
  \bibfield  {author} {\bibinfo {author} {\bibfnamefont {P.~E.}\ \bibnamefont
  {Kornilovitch}},\ }\href {\doibase 10.1103/PhysRevLett.81.5382} {\bibfield
  {journal} {\bibinfo  {journal} {Phys. Rev. Lett.}\ }\textbf {\bibinfo
  {volume} {81}},\ \bibinfo {pages} {5382} (\bibinfo {year}
  {1998})}\BibitemShut {NoStop}%
\bibitem [{\citenamefont {Rubtsov}\ \emph {et~al.}(2005)\citenamefont
  {Rubtsov}, \citenamefont {Savkin},\ and\ \citenamefont
  {Lichtenstein}}]{Rubtsov05}%
  \BibitemOpen
  \bibfield  {author} {\bibinfo {author} {\bibfnamefont {A.~N.}\ \bibnamefont
  {Rubtsov}}, \bibinfo {author} {\bibfnamefont {V.~V.}\ \bibnamefont {Savkin}},
  \ and\ \bibinfo {author} {\bibfnamefont {A.~I.}\ \bibnamefont
  {Lichtenstein}},\ }\href {\doibase 10.1103/PhysRevB.72.035122} {\bibfield
  {journal} {\bibinfo  {journal} {Phys. Rev. B}\ }\textbf {\bibinfo {volume}
  {72}},\ \bibinfo {pages} {035122} (\bibinfo {year} {2005})}\BibitemShut
  {NoStop}%
\bibitem [{\citenamefont {{Scarselli}}\ \emph {et~al.}(2009)\citenamefont
  {{Scarselli}}, \citenamefont {{Gori}}, \citenamefont {{Tsoi}}, \citenamefont
  {{Hagenbuchner}},\ and\ \citenamefont {{Monfardini}}}]{Scarselli08}%
  \BibitemOpen
  \bibfield  {author} {\bibinfo {author} {\bibfnamefont {F.}~\bibnamefont
  {{Scarselli}}}, \bibinfo {author} {\bibfnamefont {M.}~\bibnamefont {{Gori}}},
  \bibinfo {author} {\bibfnamefont {A.~C.}\ \bibnamefont {{Tsoi}}}, \bibinfo
  {author} {\bibfnamefont {M.}~\bibnamefont {{Hagenbuchner}}}, \ and\ \bibinfo
  {author} {\bibfnamefont {G.}~\bibnamefont {{Monfardini}}},\ }\href {\doibase
  10.1109/TNN.2008.2005605} {\bibfield  {journal} {\bibinfo  {journal} {IEEE
  Transactions on Neural Networks}\ }\textbf {\bibinfo {volume} {20}},\
  \bibinfo {pages} {61} (\bibinfo {year} {2009})}\BibitemShut {NoStop}%
\end{thebibliography}
\end{document}